\documentclass[reprint,
superscriptaddress,
nofootinbib,
 amsmath,amssymb,
 aps,
pra,
]{revtex4-2}

\usepackage{graphicx}
\usepackage{dcolumn}
\usepackage{bm}

\usepackage{mathtools,physics,bbm,euscript,mathrsfs,enumitem,nicefrac,amsthm,amsfonts,graphicx,booktabs,MnSymbol,eczoo}

\usepackage[unicode]{hyperref}
\hypersetup{
    colorlinks,
    linkcolor={blue!80!black},
    citecolor={blue!80!black},
    urlcolor={blue!80!black}
}

\newtheorem{theorem}{Theorem}[section]
\newtheorem{lemma}[theorem]{Lemma}

\newtheorem{proposition}[theorem]{Proposition}
\newtheorem{corollary}[theorem]{Corollary}
\newtheorem{construction}{Construction}[section]
\theoremstyle{remark}

\newtheorem{definition}{Definition}[section]
\newtheorem{example}{Example}

\newcommand\nc\newcommand
\nc\bfa{{\boldsymbol a}}\nc\bfA{{\boldsymbol A}}\nc\cA{{\EuScript A}}
\nc\bfb{{\boldsymbol b}}\nc\bfB{{\boldsymbol B}}\nc\cB{{\EuScript B}}
\nc\bfc{{\boldsymbol c}}\nc\bfC{{\boldsymbol C}}\nc\cC{{\mathscr C}}
\nc\bfd{{\boldsymbol d}}\nc\bfD{{\boldsymbol D}}\nc\cD{{\mathscr D}}
\nc\bfe{{\boldsymbol e}}\nc\bfE{{\boldsymbol E}}\nc\cE{{\EuScript E}}
\nc\bff{{\boldsymbol f}}\nc\bfF{{\boldsymbol F}}\nc\cF{{\mathscr F}}
\nc\bfg{{\boldsymbol g}}\nc\bfG{{\boldsymbol G}}\nc\cG{{\EuScript G}}
\nc\bfh{{\boldsymbol h}}\nc\bfH{{\boldsymbol H}}\nc\cH{{\mathcal H}}
\nc\bfi{{\boldsymbol i}}\nc\bfI{{\boldsymbol I}}\nc\cI{{\mathcal I}}
\nc\bfj{{\boldsymbol j}}\nc\bfJ{{\boldsymbol J}}\nc\cJ{{\EuScript J}}
\nc\bfk{{\boldsymbol k}}\nc\bfK{{\boldsymbol K}}\nc\cK{{\EuScript K}}
\nc\bfl{{\boldsymbol l}}\nc\bfL{{\boldsymbol L}}\nc\cL{{\EuScript L}}
\nc\bfm{{\boldsymbol m}}\nc\bfM{{\boldsymbol M}}\nc\cM{{\EuScript M}}
\nc\bfn{{\boldsymbol n}}\nc\bfN{{\boldsymbol N}}\nc\cN{{\EuScript N}}
\nc\bfo{{\boldsymbol o}}\nc\bfO{{\boldsymbol O}}\nc\cO{{\EuScript O}}
\nc\bfp{{\boldsymbol p}}\nc\bfP{{\boldsymbol P}}\nc\cP{{\EuScript P}}
\nc\bfq{{\boldsymbol q}}\nc\bfQ{{\boldsymbol Q}}\nc\cQ{{\mathcal Q}}
\nc\bfr{{\boldsymbol r}}\nc\bfR{{\boldsymbol R}}\nc\cR{{\EuScript R}}
\nc\bfs{{\boldsymbol s}}\nc\bfS{{\boldsymbol S}}\nc\cS{{\EuScript S}}
\nc\bft{{\boldsymbol t}}\nc\bfT{{\boldsymbol T}}\nc\cT{{\EuScript T}}
\nc\bfu{{\boldsymbol u}}\nc\bfU{{\boldsymbol U}}\nc\cU{{\EuScript U}}
\nc\bfv{{\boldsymbol v}}\nc\bfV{{\boldsymbol V}}\nc\cV{{\mathscr V}}
\nc\bfw{{\boldsymbol w}}\nc\bfW{{\boldsymbol W}}\nc\cW{{\mathscr W}}
\nc\bfx{{\boldsymbol x}}\nc\bfX{{\boldsymbol X}}\nc\cX{{\EuScript X}}
\nc\bfy{{\boldsymbol y}}\nc\bfY{{\boldsymbol Y}}\nc\cY{{\mathscr Y}}
\nc\bfz{{\boldsymbol z}}\nc\bfZ{{\boldsymbol Z}}\nc\cZ{{\EuScript Z}}
\nc{\1}{{\mathbbm{1}}}

\nc{\remove}[1]{}

\DeclareSymbolFont{bbold}{U}{bbold}{m}{n}
\DeclareSymbolFontAlphabet{\mathbbold}{bbold}

\DeclareMathOperator{\supp}{supp}

\newcommand{\nb}{{\bar{n}}}
\nc\reals{{\mathbb R}}
\nc{\ff}{{\mathbb F}}
\nc{\PP}{{\mathbb P}}
\nc{\complex}{{\mathbb C}}

\usepackage{xcolor}

\allowdisplaybreaks

\nc{\lket}[1]{{\left\vert{#1}\right\rangle}}

\begin{document}

\preprint{APS/123-QED}

\title{Class of codes correcting absorptions and emissions}

\author{Arda Aydin}
\affiliation{ISR and Dept. of ECE, University of Maryland, College Park, MD 20742
}
\author{Alexander Barg}%
\affiliation{ISR and Dept. of ECE, University of Maryland, College Park, MD 20742
}
\affiliation{ QuICS, NIST/UMD}




\date{\today}

\begin{abstract}
We construct a general family of quantum codes that protect against all emission, absorption, dephasing, and raising/lowering errors up to an arbitrary fixed order. Such codes are known in the literature as absorption-emission (AE) codes. We derive simplified error correction conditions for a general AE code and show that any permutation-invariant code that corrects $\le t$ errors can be mapped to an AE code that corrects up to order-$t$ transitions. Carefully tuning the parameters of permutationally invariant codes, we
construct several examples of efficient AE codes, hosted in systems with low total angular momentum. 
Our results also imply that spin codes can be mapped to AE codes, enabling us to characterize logical operators for certain subclasses of such codes.
\end{abstract}

\maketitle

\section{\label{sec:Introduction} Introduction}

While the prevailing research theme in quantum information processing focuses on encoding 
and manipulating information in two-state of simple systems such as atoms, photons, or electrons,  recent 
technological developments have supported devoting attention to storing quantum information
in more complex systems such as molecules, beginning with diatomic molecules \cite{yan2013observation, ni2018dipolar,Albert2020molecule,furey2024strategiesimplementingquantumerror,sharma2024quantumerrorcorrectionunresolvable}. In particular, \cite{Albert2020molecule} proposed information encoding that supports protection against noise that causes small changes in the angular position or momentum. However,  
it was later shown in \cite{AEcodes} that this noise model does not adequately describe the real physical noise in more general systems. At the same time, such systems are difficult to create
experimentally since they rely on complex superpositions and require systems with high total angular momentum. 

Pursuing this line of thought, \cite{AEcodes} recently studied information processing in systems with error processes
governed by photon absorption, emission, and the Zeeman interaction, which represent common sources of noise in molecular systems.
 They have further proposed to study codes for 
protection against this type of noise, calling them absorption-emission or \eczoo[{AE codes}]{ae} (\cite{AEcodes} uses the ligature \AE\ as the
code name). The theory put forward in \cite{AEcodes} applies to any system that admits multiple $(2J+1)$-dimensional irreducible
representations of the total $SU(2)$ angular momentum $J$. The state space considered for such systems is spanned by the 
states of the form $\ket{J,m}$, where $m$ varies from $-J$ to $J$ and represents the $z$-axis projection of the spin.
The noisy channel acting on the system encoding information has the Kraus representation formed by the operators 
that change the values of $J$ and $m$ by one unit in both directions (or leave them intact). See Eqns.~\eqref{eq:DefErrorSet}-\eqref{eq:DefErrorSet1} for a formal description, which expresses the error operators as combinations of the Kraus
operators controlled by the Clebsch-Gordan coefficients.

The theory developed in \cite{AEcodes} motivated its authors to study constructions of codes that protect against the AE noise. Subsequently, the authors of \cite{furey2024strategiesimplementingquantumerror} studied implementation strategies for such codes in linear molecules as well as an approximate version of AE codes. 
The explicit form of the noise operators is technically involved and complicates verifying
the Knill-Laflamme (KL) conditions for error correction. To overcome this difficulty, both \cite{AEcodes} and 
\cite{furey2024strategiesimplementingquantumerror} considered codes with basis states separated by several units of the $z$-component, so that unit steps of $m$ in either
direction do not move codewords to codewords. This ansatz ``diagonalizes'' the KL conditions, enabling construction of explicit codes. At the same time, this construction approach
limits the search space for good codes, often resulting in codes with higher total
angular momentum. Our examples below confirm this intuition.
To lift these restrictions, in this work we derive simplified error correction conditions that do not rely on any separation between basis states, which leads us to  make a connection between AE codes and some other code families, notably \eczoo[permutation-invariant codes]{combinatorial_permutation_invariant} \cite{ruskai-polatsek,ouyangPI,aydin2023family}. 

We take our motivation from a remark in \cite{AEcodes}, which observed that their construction of AE codes can be relevant to code families protecting against other types of noise. To quote its authors, ``Using the ansatz consisting of identical spacing between any neighboring pair of states [...] recovers [...] the permutation-invariant GNU codes \cite{ouyangPI}.'' In this paper, we propose to construct AE codes
as images of permutation-invariant codes, whose basis states are linear combinations of Dicke states, making 
them invariant under all permutations of the qubits. Our starting point is the set of necessary and sufficient conditions for a code to be stabilized by the symmetric group, expressed as a set of equation satisfied by the coefficients of the linear combinations of Dicke states (see \cite{aydin2023family} or Theorems~\ref{theorem:Main} and 
\ref{thm:PIConditions} below in this paper).
Mapping the Dicke states onto $\ket{J,m}$ systems suggests a way of obtaining AE codes, although verifying the KL conditions still presents a challenge. Taking up this task, we work with explicit form of the Clebsch-Gordan
coefficients coupled with the basis states of permutation-invariant codes of \cite{aydin2023family}. Somewhat surprisingly, we are able to show that if the KL conditions of \cite{aydin2023family} hold true, then so do the KL conditions for AE codes. A priori this is unexpected because we have moved from one type of systems to a rather different kind, and there is little hope that the error-correcting conditions will transfer to the new setting. 

As a result, we are able to obtain a large class of AE codes relying on available families of permutation-invariant codes. These codes are relatively well researched, and we can use existing proposals to 
construct AE codes. In particular, relying on a large family of permutation-invariant codes constructed recently in \cite{aydin2023family}, we obtain a general family of AE codes. Specific examples constructed in the paper yield codes hosted by systems with the lowest known angular momentum among all codes with comparable noise resilience properties. These 
results contribute to the practical utility of the AE code constructions since molecules with lower angular momentum are more stable and therefore easier to manipulate.
Moreover, since there exist permutation-invariant codes that encode multiple qubits of information \cite{ouyangQudit}, we
are able to leverage these constructions to obtain explicit AE codes of higher dimension. Along the way, we also 
lift the spacing anzatz of \cite{AEcodes}, adding flexibility to the AE code constructions.

We further develop the connection described above by using the fact, established in \cite{gross2,kubischta2023notsosecret}, that \eczoo[spin codes]{j_gross} of \cite{gross} can be mapped to permutation-invariant codes. Combining this observation with our mapping from the latter class to AE codes, we observed that spin codes can be mapped to AE codes with similar error correction properties.
This connection enables us to construct AE codes admitting specific subgroups of $SU(2)$ as logical unitary operators, as well as other new AE codes.

\section{Preliminaries}
A quantum error correction code is a subspace of the Hilbert space of a physical system. An $n$-dimensional code is defined by its basis $\{\ket{\bfc_i} : i\in\{0,1,\ldots,n-1\}\}$, and a set of errors $\cE=\{\hat{E}_a\}$ is correctable if
\begin{align*}
    \bra{\bfc_i}\hat{E}_a^\dagger\hat{E}_b\ket{\bfc_j}=\delta_{i,j}g_{ab} \quad\quad \text{(KL conditions)}
\end{align*}
holds for all $\hat{E}_a,\hat{E}_b\in \cE$ \cite{knill,knill2}. Here, $g_{ab}$ are complex coefficients
and $\delta_{i,j}$ is the Kronecker delta. Similarly, a set of errors is detectable if
\begin{align*}
    \bra{\bfc_i}\hat{E}_a\ket{\bfc_j}=\delta_{i,j}g_{a}
\end{align*}
holds for all $\hat{E}_a\in \cE$.
In this paper, we consider a system with the total angular momentum $J=n/2$, where $n$ is a positive integer. The code space we examine is a subspace of a Hilbert space spanned by the states $ \{\ket{J,m} : m\in\{ -J,-J+1,\ldots,J-1,J \} \} $. Throughout this paper, we will study two-dimensional codes with the basis
\begin{align}\label{eq:DefCodeSpace}
    \ket{\bfc_0} = \sum_{j=0}^n \alpha_j\ket{n/2,j-n/2},\quad \ket{\bfc_1} = \sum_{j=0}^n \beta_j\ket{n/2,j-n/2}, 
\end{align}
where $\alpha_j,\beta_j$ are complex coefficients. 
We say that the code can correct (detect) up to {\em order-$t$ transitions} if it corrects (resp., detects) the set of errors given by \cite{AEcodes}
\begin{align}\label{eq:DefErrorSet}
    \cE_t = \{\hat{E}_{\delta m}^{r,\delta J}: |\delta J|\leq r\leq t, |\delta m|\leq r\leq t \}.
\end{align}
Here the error operators have the form
\begin{align}\label{eq:DefErrorSet1}
    \hat{E}_{\delta m}^{r,\delta J} \propto \sum_{m=-J}^{J} C^{J+\delta J,m+\delta m}_{J,m;r,\delta m}\ket{J+\delta J,m+\delta m}\bra{J,m},
\end{align}
where $C^{J+\delta J,m+\delta m}_{J,m;r,\delta m}$ are the Clebsch-Gordan coefficients that express the coupled angular momentum in terms of the uncoupled basis (see e.g., \cite[Sec.3.6]{sakurai} or \cite{ClebcshBook} for an introduction to the theory of angular momentum).
The main challenge in constructing the codes lies in properly assigning the coefficients $\alpha_j,\beta_j$ in the code construction \eqref{eq:DefCodeSpace} to fulfil the KL conditions for the set of errors $\cE_t$ defined above.

\subsection{Spin codes} 
Spin codes  are designed to protect the information encoded in a spin-$J$ system \cite{gross}. To define it, denote by $\ket{J,m}$ an eigenstate of the $z$-component of the angular momentum operator $J_z$, where $J$ an integer or a half-integer. The Hilbert space of a spin-$J$ system has the basis $\{\ket{J,m},m\in\{-J,-J+1,\dots,J\}\}$.
 A $K$-dimensional \textit{spin code} is simply a $K$-dimensional subspace of a spin-$J$ system.  Throughout this paper, we study $2$-dimensional codes that have a basis of the form
\begin{align*}
    \ket{\bfc_0} = \sum_{m=-J}^J\alpha_m\ket{J,m} \quad \text{and} \quad \ket{\bfc_1} =  \sum_{m=-J}^J\beta_m\ket{J,m},
\end{align*}
where $\alpha_m,\beta_m$ are complex coefficients.
Such codes protect the encoded information against small-order isotropic rotations. A convenient description of the error
set is obtained relying on the Wigner-Eckart theorem \cite[Sec.~3.11]{sakurai} and can be written in terms of spherical tensors \cite{gross2}:
\begin{align}\label{eq:DefErrorSetSpin}
    \cE_t^\prime = \{\hat{E}_{\delta m}^{r}: |\delta m|\leq r\leq t \},
\end{align}
where
\begin{align*}
     \hat{E}_{\delta m}^{r} \propto \sum_{m=-J}^{J} C^{J,m+\delta m}_{J,m;r,\delta m}\ket{J,m+\delta m}\bra{J,m}.
\end{align*}
We say that a spin code can correct (detect) random rotation errors of order-$t$ if it corrects (detects) the errors from the set defined in \eqref{eq:DefErrorSetSpin}. Note that an AE code can be realized as a spin code since the error set in \eqref{eq:DefErrorSetSpin} forms a subset of the error set defined in \eqref{eq:DefErrorSet}. In principle, this error
correction condition can be used to define the distance of spin codes by postulating that the distance should be greater than
$t$ in the definition of $\cE_t'$, but we will not use this concept in our results. An analogous remark could be also made for AE codes.

Spin codes were introduced in \cite{gross}, which also constructed a family of codes 
correcting first-order random rotations. The primary goal in \cite{gross} was to design codes in which single-qubit Clifford operations are realized as logical unitaries. Subsequently, the authors of \cite{gross2} studied the problem of encoding a qubit 
into a collection of spin systems, extending the result of \cite{gross} for single-spin codes. 
Recently, \cite{kubischta2023notsosecret} studied a subclass of spin codes with binary dihedral symmetry.

\subsection{Permutation-invariant quantum codes} The AE codes that we construct are obtained by a mapping from
permutation-invariant codes. We start with the definition of this code family, which is due to \cite{ruskai-polatsek}.
An $n$-qubit permutation-invariant quantum code $\cC$ is a multi-qubit quantum code that is 
stabilized by the symmetric group $S_n$, i.e., for all $\ket{\psi}\in \cC$
    \begin{align*}
        g\ket{\psi} = \ket{\psi},\quad \text{for all $g\in S_n$}.
    \end{align*}
Such codes are conveniently described using Dicke states \cite{Dicke1,Dicke2,Dicke3}. A {\em Dicke state} $\ket{D^n_w}$ is a quantum state that is a linear combination of all $n$-qubit states with the same excitation. In other words,
\begin{align*}
    \ket{D^n_w}= \frac{1}{\sqrt{\binom{n}{w}}}\sum_{\substack{\bfx\in\{0,1\}^n\\|\bfx|=w}}\ket{\bfx},
\end{align*}
where $|\bfx|$ is the Hamming weight of the binary string $\bfx$. In this paper, we will primarily study two-dimensional permutation-invariant codes, which can be defined as follows:
\begin{definition}\label{def:PICodes}
    A two-dimensional permutation-invariant code is a subspace defined by basis vectors of the form
    \begin{align}\label{eq:PIDef}
        \ket{\bfc_0} = \sum_{j=0}^n\alpha_j\ket{D^n_j}\quad\text{and}\quad \ket{\bfc_1}=\sum_{j=0}^n\beta_j\ket{D^n_j},
    \end{align}
    where $\alpha_j,\beta_j\in\mathbb{C}$.
\end{definition}
Permutation-invariant codes were introduced by Pollatsek and Ruskai in their works \cite{ruskai-polatsek,ruskaiExchange}. The main motivation behind these works was to construct quantum codes that are 
resilient against spin exchange errors, which occur due to the nature of many-body systems formed of identical 
particles. Pursuing this goal, \cite{ruskai-polatsek,ruskaiExchange} introduced a family of $2$-dimensional codes that can correct a single error. 
Subsequently, Ouyang \cite{ouyangPI} observed that codewords of any permutation-invariant code form ground states of a 
Heisenberg ferromagnet. Motivated by this, he introduced the first family of permutation-invariant codes capable of correcting arbitrary number of errors. Ouyang's codes are defined by three integer parameters, $g$, $n$, and $u$, 
and are hence known as \eczoo[GNU codes]{gnu_permutation_invariant}. Codes from this family encode a single qubit into $gnu$ physical qubits.

Thereafter, Ouyang studied permutation-invariant codes in higher dimensions \cite{ouyangHigherDimensions, ouyangQudit}. He showed in \cite{ouyangQudit}, Theorem 6.7 that there exists a $k$-dimensional permutation-invariant code that can correct $t$ errors and has length $(k-1)(2t+1)^2$. In the same paper, he introduced an explicit method for constructing these codes.

More recently, the authors of \cite{aydin2023family} introduced necessary and sufficient conditions for a  permutation-invariant code to correct $t$ errors (this result is cited as Theorem \ref{thm:PIConditions} in the Appendix). They further introduced a new family of codes that can correct an arbitrary number of errors. The shortest code from this family has length $(2t+1)^2-2t$,
which compares favorably with previous construction. Namely, the shortest $t$-error-correcting codes of \cite{ouyangPI} have length $(2t+1)^2$ (and the shortest single-error-correcting code in \cite{ruskai-polatsek} has length 7). 

In this paper, we show that permutation-invariant codes with distance $d=2t+1$ can be mapped to AE codes that can correct up to order-$t$ transitions. This enables us to construct more efficient codes by using the previously known permutation-invariant codes. We also use the correspondence between spin codes and the permutation-invariant codes, previously established in \cite{gross2,kubischta2023notsosecret}. Combining this with our result, we conclude that any spin code can be realized as an AE code with similar error correction properties.

\section{Code Construction}

\subsection{Conditions for error correction}
The error correction conditions for AE codes introduced in \cite{AEcodes} are based on the assumption of a sufficiently large 
spacing between excitation levels. Its authors only consider subspaces formed as superpositions of basis states $\{\ket{J,m}\}$
that are separated by at least $2t+1$ units in $m$ (see Appendix~\ref{sec: example} for the detais of their approach). Due to this specific choice of the codespace, the off-diagonal KL 
conditions are automatically satisfied. At the same time, as noted above, this restriction on the excitation levels of the basis states may force us to construct codes with larger total angular momentum. The {\em main technical result} of this work yields generalized and simplified error correction conditions that do not require any assumptions about basis states. It is stated in the next theorem, whose proof is deferred to Appendix~\ref{sec:main thm proof}.

\begin{theorem}\label{theorem:Main}
    Let $\cC$ be an AE quantum error correction code with real coefficients $\alpha_j$ and $\beta_j$, $j=1,\ldots,n$ as given in \eqref{eq:DefCodeSpace}. If equalities
    \begin{align*}
        &\text{\rm(C1)}\quad \sum_{j=0}^n\alpha_j\beta_j = 0;\\
        &\text{\rm(C2)}\quad \sum_{j=0}^n\alpha_j^2 = \sum_{j=0}^n\beta_j^2 = 1;\\
        &\text{\rm(C3) \;For all $ 0\leq  a,b  \leq t^\prime$},\\
        & \hspace*{0.5in}\sum_{j=0}^n\frac{\binom{n-2t}{j}}{\sqrt{\binom{n}{j+a} \binom{n}{j+b}}} \alpha_{j+a}\beta_{j+b}=0;\\
        &\text{\rm(C4) \;For all $ 0\leq  a,b  \leq t^\prime$},\\
        & \hspace*{0.5in}\sum_{j=0}^n\frac{\binom{n-2t}{j}}{\sqrt{\binom{n}{j+a} \binom{n}{j+b}}}\left(\alpha_{j+a}\alpha_{j+b}-\beta_{j+a}\beta_{j+b}\right)=0
    \end{align*}
    hold for $t^\prime=2t$, then code $\cC$ corrects order-$s$ transitions for all $s\le t$. If they hold for $t^\prime=t$, then it detects up to order-$t$ transitions. 
\end{theorem}
\noindent{\em Remark:} We assume by definition that $\frac{\alpha_k}{\sqrt{\binom{n}{k}}}=\frac{\beta_k}{\sqrt{\binom{n}{k}}}=0$ if $k>n$.

Solving the set of equations introduced in conditions (C1)-(C4) we can construct AE codes that can correct up to order-$t$ transitions. In Appendix~\ref{sec: example}  we give details of the AE code construction implied by Theorem \ref{theorem:Main}, noting that this recovers the construction obtained earlier in \cite{AEcodes} as a particular case.

The following proposition is based on the connection to permutation-invariant codes
established by Theorems \ref{theorem:Main} and \ref{thm:PIConditions}.

\begin{proposition}\label{prop:PICodes}
    Let $\ket{D^n_j}$ be a Dicke state. Define the mapping 
    \begin{align}\label{eq:Gmap}
        \ket{D^n_j} \stackrel{e}{\longmapsto} \ket{n/2,j-n/2}.
    \end{align}
    Then, any $k$-dimensional $t$-error correcting permutation-invariant quantum code with real coefficients and length $n$ is mapped by $e$ to a $k$-dimensional AE code with total angular momentum $J=n/2$ that can correct up to order-$t$ transitions, and can detect up to order-$2t$ transitions.
\end{proposition}    
This result leads to the following useful observation. The authors of \cite{AEcodes} noticed
that the code family they constructed to correct up to order-$t$ transitions recovers the GNU codes of Ouyang \cite{ouyangPI}. Proposition \ref{prop:PICodes} implies that {\em any} permutation-invariant code capable of correcting $t$ errors can be mapped to an AE code that corrects up to order-$t$ transitions. Therefore, leveraging the existing literature on permutation-invariant codes, we can construct more efficient AE codes than previously known. With this in mind, using Construction \ref{constructionGmdelta} along with the mapping defined in \eqref{eq:Gmap}, we introduce the following family of AE codes.

\begin{construction}\label{cons:gmdeltaCodes}
    Let $g,m,\delta $ be nonnegative integers, and $\epsilon\in\{-1,+1\}$. Define an AE code $\cQ_{g,m,\delta,\epsilon}$ with the logical computational basis
    \begin{align*}
        &\ket{\bfc_0} =\sum_{\substack{\text{$l$ {\rm even}}\\0\leq l \leq m}} \gamma b_l\;\ket{n/2,gl-n/2} + 
        \sum_{\substack{\text{$l$ {\rm odd}}\\0\leq l \leq m}} \gamma b_l\;\ket{n/2,n/2-gl},\\
        &\ket{\bfc_1} = \sum_{\substack{\text{$l$ {\rm odd}}\\0\leq l \leq m}} \gamma b_l\;\ket{n/2,gl-n/2} +\epsilon 
        \sum_{\substack{\text{$l$ {\rm even}}\\0\leq l \leq m}} \gamma b_l\;\ket{n/2,n/2-gl},
    \end{align*}
    where $ n=2gm+\delta+1, $
    $
    b_l:=\sqrt{{\binom{m}{l}}/{\binom{n/g-l}{m+1} }},
    $
and $ \gamma =  \sqrt{\binom{n/(2g)}{m} \frac{n-2gm}{g(m+1)} } $ is the normalizing factor\footnote{Binomial coefficients with non-integer entries are defined in Eq.~\eqref{eq: bin} below.}.
\end{construction}

In the following theorem we explicitly state the error correction properties of codes in Construction \ref{cons:gmdeltaCodes}.
\begin{theorem}\label{theorem:ErrorCorrection}
    Let $m,g,\delta,t$ be nonnegative integers. If $m\geq t$, $\delta\geq 2t$, and
    \begin{align*}
        (g\geq 2t,\epsilon=-1) \quad\text{or}\quad (g\geq 2t+1,\epsilon=+1),
    \end{align*}
    then the code $\cQ_{g,m,\delta,\epsilon}$ is hosted in a system with total angular momentum $J=\frac{2gm+\delta+1}{2}$ and corrects up to order-$t$ transitions.
\end{theorem}
\begin{proof} By combining Theorem \ref{theoremGMdelta} with Proposition \ref{prop:PICodes}. 
\end{proof}

The following two examples are constructed using the result of Theorem \ref{theorem:ErrorCorrection}.
\begin{example}\label{example:7qubit}
    Let $g=2,m=1,\delta=2$, and $\epsilon=-1$. Then the code $\cQ_{2,1,2,-}$ defined via its basis
    \begin{equation}\label{eq:7qubitCode}
\left.
    \begin{aligned}
        \ket{\bfc_0} &= \sqrt{\frac{3}{10}}\, \lket{\frac{7}{2},-\frac{7}{2}} + \sqrt{\frac{7}{10}}\, \lket{\frac{7}{2},\frac{3}{2}} \\
        \ket{\bfc_1} &= \sqrt{\frac{7}{10}}\, \lket{\frac{7}{2},-\frac{3}{2}} - \sqrt{\frac{3}{10}}\, \lket{\frac{7}{2},\frac{7}{2}}
    \end{aligned}\;\;\right\} 
    \end{equation}
    can correct a single transition error.  Note that this requires a system with total angular momentum $J=7/2$, which is less than required by the best code introduced in \cite{AEcodes}. 
\end{example}
\begin{example}
     Suppose $g=4,m=2,\delta=4$, and $\epsilon=-1$. Then, the code $\cQ_{4,2,4,-}$ defined via its basis
     \begin{align*}
         &\ket{\bfc_0} = 
         \sqrt{\frac{5}{68}}\, \lket{\frac{21}{2},-\frac{21}{2}}+ \sqrt{\frac{7}{12}}\, \lket{\frac{21}{2},-\frac{5}{2}} +  \sqrt{\frac{35}{102}}\, \lket{\frac{21}{2},\frac{13}{2}},\\
         &\ket{\bfc_1} = \sqrt{\frac{35}{102}}\, \lket{\frac{21}{2},
         -\frac{13}{2}}- \sqrt{\frac{7}{12}}\, \lket{\frac{21}{2},\frac{5}{2}} -  \sqrt{\frac{35}{102}}\, \lket{\frac{21}{2},\frac{21}{2}}         
     \end{align*}
corrects single and double transition errors. 
     Note that this requires a system with total angular momentum $J=21/2$, which again outperforms the best
     previously known AE codes capable of correcting up to double transition errors.
\end{example}

\subsection{Our construction and prior results}
Note that the code family defined in Construction \ref{cons:gmdeltaCodes} has the property that the values of $z$-axis projection of the momentum are symmetric about $n/2$. The authors of \cite{AEcodes} call such codes \textit{counter-symmetric}. The most efficient codes from this family that can correct up to order-$t$ transitions are obtained by setting $g=2t$, $m=t$, $\delta=2t$, and $\epsilon=-1$. These codes can be hosted in a system with total angular momentum $J=(2t+1)^2/2-t$. Previously known AE codes 
with comparable correction properties (up to order-$t$ transitions) require a total angular momentum of at least $J=(2t+1)^2/2$, which is $t$ more than for our code family. Furthermore, the code $\cQ_{g,\frac{m-1}{2},g-1,+}$ coincides with the counter-symmetric codes introduced for higher-order transitions in \cite{AEcodes} (See Proposition 5.4 in \cite{aydin2023family}). Hence, our codes not only generalize the previously known families of codes for higher-order transitions, but also introduce many new AE codes that require a physical system with lower total angular momentum.   

In \cite{aydin2023family}, we also introduced another code family by generalizing Pollatsek and Ruskai's single-error-correcting permutation-invariant codes \cite{ruskai-polatsek}. In particular, we found a $((19,2,5))$ permutation-invariant code by solving the set of equations derived from conditions (C1)-(C4) for a specific orientation of the support of the codewords. By Proposition \ref{prop:PICodes}, we conclude that there exists an AE code capable of correcting up to double-transition errors, with a host system of total angular momentum $J=19/2$. The best previously known code of this kind is hosted by a system with total angular momentum $J=25/2$, so our construction yields a more efficient solution.
\subsection{Encoding many qubits}
Proposition \ref{prop:PICodes} also enables construction of AE codes in higher dimensions. Noting this, we
leverage known results on many-qubit permutation-invariant codes. One such construction was obtained in \cite[Thm.~5.2]{ouyangQudit}, which 
introduced an explicit way to construct $t$-error correcting, $k$-dimensional permutation-invariant code with length $n\geq (k-1)(2t+1)^2$. By Proposition \ref{prop:PICodes}, one can therefore construct $k$-dimensional AE codes that require a total angular momentum $J=\frac{(k-1)(2t+1)^2}{2}$. As a proof of concept, in the following example we construct an AE code with more than two logical basis states: The AE code with basis states
\begin{align*}
    &\ket{\bfc_0}= \frac{1}{4}\, \lket{\frac{27}{2},-\frac{27}{2}} + \frac{\sqrt{12}}{4}\, \, \lket{\frac{27}{2},-\frac{3}{2}} + \frac{\sqrt{3}}{4}\,\, \lket{\frac{27}{2},\frac{21}{2}},\\
    &\ket{\bfc_1}=\frac{\sqrt{3}}{4}\, \lket{\frac{27}{2},-\frac{21}{2}} + \frac{\sqrt{12}}{4}\, \lket{\frac{27}{2},\frac{3}{2}}+\frac{1}{4}\, \lket{\frac{27}{2},\frac{27}{2}},\\
    &\ket{\bfc_2}=\frac{\sqrt{6}}{4}\, \lket{\frac{27}{2},-\frac{15}{2}} + \frac{\sqrt{10}}{4}\, \lket{\frac{27}{2},\frac{9}{2}},\\
    &\ket{\bfc_3}= \frac{\sqrt{10}}{4}\, \lket{\frac{27}{2},-\frac{9}{2}} + \frac{\sqrt{6}}{4}\, \lket{\frac{27}{2},\frac{15}{2}}
\end{align*}
encodes two bits of information into a system with total angular momentum $J=27/2$. It corrects up to a single transition error and detects up to order-two transitions. This result is obtained by applying the mapping $e$ to the permutation-invariant code from Example 6.5 in \cite{ouyangQudit}.
To the best of our knowledge, all previously known AE codes encode a single qubit of information, so this example seems
to be the first one of AE codes in higher dimensions.  

\subsection{Other constructions} Reliance on permutation-invariant codes is not the only possible way to construct AE codes: another way to obtain them is starting from the known 
constructions of spin codes.  Paper \cite{gross2} was the first to observe the relationship between spin codes and permutation-invariant codes (see their Section VII.). Subsequently, \cite{kubischta2023notsosecret} presented a rigorous argument connecting spin codes and permutation-invariant codes with the same distance. The authors called the next lemma {\em Dicke bootstrap}.
\begin{lemma}[\cite{kubischta2023notsosecret}, Lemma 2]\label{lemma:SpinCodesToPICodes}
    Define the mapping
    \begin{align}\label{eq:Hmap}
        \ket{J,m} \stackrel{h}{\longmapsto} \ket{D^{2J}_{J-m}}.
    \end{align}
 Then a spin code with angular momentum $J$ that corrects random rotations of order (up to) $t$ corresponds under mapping $h$
 to a permutation-invariant code of length $2J$ that corrects $t$ errors. 
\end{lemma}
Relying on this result, in the next proposition we establish a relationship between spin codes and AE codes.
\begin{proposition}\label{prop:SpinCodes}
    Define the mapping $f=e\circ h$ 
    \begin{align*}
        \ket{J,m} \stackrel{h}{\longmapsto} \ket{D^{2J}_{J-m}} \stackrel{e}{\longmapsto} \ket{J,m},
    \end{align*}
   where $h$  is defined in \eqref{eq:Hmap} and $e$ is the mapping in \eqref{eq:Gmap}.  Then a spin code that corrects random rotation errors of order-$t$ with a spin-$J$ system is mapped by $f$ to an AE code hosted in a system with the same total angular momentum $J$. This code can correct up to order-$t$ transitions and detect up to order-$2t$ transitions.
\end{proposition}
\begin{proof} Follows by combining Lemma \ref{lemma:SpinCodesToPICodes} and Proposition \ref{prop:PICodes}.
\end{proof}
As observed in \cite{AEcodes}, an AE code that corrects up to order-$t$ transitions can be viewed as a spin code \cite{gross} that corrects the same order of isotropic rotations. We showed in Proposition \ref{prop:SpinCodes} that a spin code can also be realized as an AE code with similar error correction properties. Hence, by Proposition \ref{prop:SpinCodes}, we can use the previously known single-spin codes to construct AE codes.

\subsection{Logical operators for AE codes} 
To describe the action of the rotation operator $R$ on a spin-$J$ system, recall that the group $SU(2)$ admits a $2J+1$-dimensional
irreducible representation by Wigner matrices $\cD_{m,m}^{j}(R)$ with matrix elements
    $$
    \cD_{m,m}^{j}(R)=\langle j,m' | \exp\Big(\frac{-i\,{\mathbf J}\cdot\hat{\mathbf n}\phi}{\hbar}\Big)|j,m\rangle,
    $$
where $\hat{\mathbf n}$ and $\phi$ define the rotation \cite[p.~196]{sakurai}.    
In other words, an element $g\in SU(2)$ acts on a spin-$J$ system via the operator 
$\cD_{m,m}^{j}(R)=\cD_{m,m}^{j}(g)$. Let $G$ be a subgroup of $SU(2)$. If the codespace of a large-spin system is preserved under $D^J(g)$ for all $g\in G$, the code is called {\em $G$-covariant} \cite{CovariantCodes1}.
Papers \cite{gross,kubischta2023notsosecret,exoticGates} studied $G$-covariant spin codes, where $G$ is a subgroup of $SU(2)$.
The specific subgroups considered there are the binary 
octahedral group $2O$, also called the Clifford group, and the binary icosahedral 
group $2I$.

In particular, the authors of \cite{gross} constructed is a $J=13/2$, $2O$-covariant spin code capable of correcting first-order random rotations. They also constructed
a $J=7/2$, $2I$-covariant spin code capable of correcting first-order random rotations. Using our Proposition \ref{prop:SpinCodes} we observe that
the first of these codes is also a $2O$-covariant AE code that corrects order-$1$ transitions. 
The second code turns to coincide with the code defined by Eq.~\eqref{eq:7qubitCode}. Again referring to Proposition \ref{prop:SpinCodes}, we conclude that it can be viewed as an AE code in a single large spin $7/2$ system that corrects a singe transition error and admits the representatives from $2I$ as logical unitaries.

Another subgroup, considered in \cite{kubischta2023notsosecret}, is the binary dihedral group $BD_{2b}$ of order $2b$,
which is a non-abelian subgroup of $SU(2)$ defined as follows:
\begin{align*}
    BD_{2b} = \left\langle X,Z,\begin{pmatrix}
        e^{-i\pi/2b} & 0 \\ 0 & e^{i\pi /2b}
    \end{pmatrix} \right\rangle.
\end{align*}
This paper further defined a family of $BD_{2b}$-covariant family of spin codes capable of correcting first-order random rotations. We observe that
this family is a particular case of Construction \ref{cons:gmdeltaCodes} once we make a proper assignment of the parameters.
We obtain the following statement.
\begin{proposition}
    For an integer $r\geq 3$, the code $\cQ_{3,1,2^r-4,+}$ is an AE code which is hosted in a system with $J=\frac{2^r+3}{2}$, is $BD_{2^r}$-covariant, correct a single transition error, and detect up to double transition errors.
\end{proposition}
\begin{proof}
The proof is immediate by observing that applying $f=e\circ h$ to the code defined in \cite[Eqns. (35), (36)]{kubischta2023notsosecret}  results in the AE code  $\cQ_{3,1,2^r-4,+}$.  
\end{proof}
For example, the AE code $Q_{3,1,4,+}$, encoded in a system with angular momentum $J=11/2$, with its basis
\begin{align*}
    &\ket{\bfc_0} = \frac{\sqrt{5}}{4}\, \lket{\frac{11}{2},-\frac{11}{2}} + \frac{\sqrt{11}}{4}\, \lket{\frac{11}{2},\frac{5}{2}},\\
    &\ket{\bfc_1} = \frac{\sqrt{11}}{4}\, \lket{\frac{11}{2},-\frac{5}{2}} + \frac{\sqrt{5}}{4}\, \lket{\frac{11}{2},\frac{11}{2}}
\end{align*}
can correct a single transition, detect up to order-$2$ transitions, and realizes any representative from the group $BD_8$ as a logical unitary.

\vfill\eject
\section{Summary and concluding Remarks}
Our main results in this paper are as follows:\\[.05in]
(1) We find a set of sufficient conditions for AE codes to correct up to order-$t$ transition
    errors for any $t\ge 1$, expressed in terms of the basis states of the code;\\[.05in]
(2)  We show that any {\em permutation-invariant} qubit code \cite{ruskai-polatsek}
that corrects up to $t$ errors can be mapped onto an AE code that corrects up to order-$t$ transitions. 
   Coupled with the general family of permutation-invariant codes constructed recently in
   \cite{aydin2023family}, this yields a large class of AE codes.   We also give examples of efficient AE codes hosted in systems with low total angular momentum compared to previous benchmarks.\\[.05in]
(3) We further show that {\em spin codes} \cite{gross,gross2} are nearly equivalent to AE
codes in terms of error correction. This enables us to construct new families of AE codes that implement logical unitaries from subgroups of $SU(2)$.\\[.0in]

We further note that the authors of \cite{AEcodes} observed that AE codes are closely related to
{\em binomial codes} \cite{michael2016new}, which form a certain class of bosonic codes. 
Uncovering the details of this equivalence
and its implications forms an interesting direction for future research.

\begin{acknowledgments}
We are grateful to Victor Albert for a useful discussion of \cite{AEcodes} and spin codes. This research was partially supported by NSF grant CCF2330909.
\end{acknowledgments}

\vfill\eject
\appendix

\onecolumngrid

\section{A family of Permutation-Invariant Codes}\label{sec:PI-codes}
 Since the construction of AE codes proceeds by a transformation from permutation-invariant codes, in this appendix we summarize results on them from \cite{aydin2023family}.

We begin with recalling the construction of permutation-invariant codes.
\begin{construction}[\cite{aydin2023family}, Construction 5.1]\label{constructionGmdelta}
 Let $g,m,\delta $ be nonnegative integers, and let $\epsilon\in\{-1,+1\}$. Define a permutation-invariant code $\cT_{g,m,\delta,\epsilon}$ via its logical computational basis
    \begin{align*}
        &\ket{\bfc_0} =\sum_{\substack{\text{$l$ {\rm even}}\\0\leq l \leq m}} \gamma b_l\ket{D^n_{gl}} + 
        \sum_{\substack{\text{$l$ {\rm odd}}\\0\leq l \leq m}} \gamma b_l\ket{D^n_{n-gl}},\\
        &\ket{\bfc_1} = \sum_{\substack{\text{$l$ {\rm odd}}\\0\leq l \leq m}} \gamma b_l\ket{D^n_{gl}} +\epsilon
        \sum_{\substack{\text{$l$ {\rm even}}\\0\leq l \leq m}} \gamma b_l\ket{D^n_{n-gl}},
    \end{align*}
    where $ n=2gm+\delta+1, $
    $
    b_l=\sqrt{{\binom{m}{l}}/{\binom{n/g-l}{m+1} }},
    $
and $ \gamma =  \sqrt{\binom{n/(2g)}{m} \frac{n-2gm}{g(m+1)} } $ is the normalizing factor. 
\end{construction}

The necessary and sufficient conditions for error correction with permutation-invariant codes parallel the results for the AE codes established in Theorem~\ref{theorem:Main}. 
\begin{theorem}[\cite{aydin2023family}, Theorem 4.1]\label{thm:PIConditions}
Let $\cC$ be a permutation-invariant code as defined in \eqref{eq:PIDef}. The code $\cC$ corrects $t$ errors if and only if the real coefficients $\alpha_j$,$\beta_j$, $j=1,2,\ldots,n$ satisfy Conditions (C1)-(C4) of Theorem~\ref{theorem:Main} with $t'=2t$.
\end{theorem}
\remove{    \begin{align*}
        &\text{\rm(C1)}\quad \sum_{j=0}^n\alpha_j\beta_j = 0;\\
        &\text{\rm(C2)}\quad \sum_{j=0}^n\alpha_j^2 = \sum_{j=0}^n\beta_j^2 = 1;\\
        &\text{\rm(C3) \;For all $ 0\leq  a,b  \leq 2t$},\\
        & \hspace*{1in}\sum_{j=0}^n\frac{\binom{n-2t}{j}}{\sqrt{\binom{n}{j+a} \binom{n}{j+b}}} \alpha_{j+a}\beta_{j+b}=0;\\
        &\text{\rm(C4) \;For all $ 0\leq  a,b  \leq 2t$},\\
        & \hspace*{1in}\sum_{j=0}^n\frac{\binom{n-2t}{j}}{\sqrt{\binom{n}{j+a} \binom{n}{j+b}}}\left(\alpha_{j+a}\alpha_{j+b}-\beta_{j+a}\beta_{j+b}\right)=0.
    \end{align*}}
It is important to note that this theorem gives a full characterization of permutation-invariant codes, justifying
our claim that {\em any} such code can be converted into an AE code with easily described parameters as detailed in
Proposition~\ref{prop:PICodes}.

Using this theorem for Construction~\ref{constructionGmdelta} presented above, \cite{aydin2023family} derived the error-correction properties of 
the codes $\cT_{g,m,\delta,\epsilon}$.
\begin{theorem}[\cite{aydin2023family}, Theorem 5.3]\label{theoremGMdelta}
   Let $ t $ be a nonnegative integer and let $ m\geq t $ and $\delta\geq 2t$. If
    $$
(   g\ge 2t, \epsilon=-1) \text{ or }(g\ge 2t+1,\epsilon=+1),
   $$
then the code $\cT_{m,l,\delta,\epsilon}$ encodes one qubit into $n = 2gm+\delta+1 $ qubits and corrects any $t$ qubit errors. 
\end{theorem}

\section{Combinatorial lemmas}\label{sec:combinatorial}

Throughout the paper we use the following standard convention for the binomial coefficients: Let $x$ be a real number and $k$ be an integer. Then,
\begin{equation}\label{eq: bin}
    \binom{x}{k} = \begin{cases}
        \frac{x(x-1)\ldots(x-k+1)}{k!}  &k>0,\\
        1  &k=0,\\
        0  &k<0.
    \end{cases}
\end{equation}
The following binomial identity will be used in the proof of Lemma \ref{lemma:Combinatorics2}.
\begin{lemma}\label{Lemma:Combinatorics1}
    Let $n,k,r,a$ be nonnegative integers such that $n\geq r\geq a$ and $n\geq k\geq a$. Then,
    \begin{equation*}
        \frac{\binom{n-r}{k-a}}{\binom{n}{k}} = \frac{\binom{n-k}{r-a}\binom{k}{a}}{\binom{n}{r}\binom{r}{a}}.
    \end{equation*}
\end{lemma}
\begin{proof} 
Rewriting the fractions, the lemma claims that 
  $$
  \binom{n}{r}\binom{n-r}{k-a}\binom{r}{a} 
  =\binom{n}{k}\binom{n-k}{r-a}\binom{k}{a}.
  $$
On the left we first choose $r$ elements out of $n$, then $a$ out of $r$, then $k-a$ out of $n-r$. On the right,
we first choose $k$ out of $n$, then $r-a$ out of the remaining $n-k$, and then $a$ out of $k$. 
Apart from the changing the order, both sides seek to construct the same groups of subsets.
%
\end{proof}

In the next lemma we quote a combinatorial identity that plays an important role in our proof. 
\begin{lemma}[\cite{CombinatorialIdentity}] \label{Theorem:Combinatorics}
     Let $a,b,c,d,e$ be nonnegative integers. Then 
    \begin{align}\label{eq:combTheorem}
        \binom{a+c+d+e}{a+c}\binom{b+c+d+e}{c+e}=\sum_i \binom{a+b+c+d+e-i}{a+b+c+d}\binom{a+d}{i+d}\binom{b+c}{i+c}.
    \end{align}
\end{lemma}
Adjusting the notation to our needs, we obtain
\begin{corollary}\label{Corr:Combinatorics}
    Let $n,l,m$ be natural numbers and $r$ be an integer. Then,
    \begin{align}\label{eq:corollary}
        \binom{n+m+r}{n+r}\binom{l+m}{r} = \sum_{i=0}^m \binom{n+l+m+i}{i}\binom{n+m}{n+i}\binom{l}{r-i}.
    \end{align}
\end{corollary}
\begin{proof}
Clearly, if $r<0$, then both the LHS and RHS of \eqref{eq:corollary} are zero. If $r\geq 0$, then \eqref{eq:corollary} is obtained from \eqref{eq:combTheorem} by setting $n \leftmapsto a, l\leftmapsto b+c, m\leftmapsto d, r\leftmapsto c, e\leftmapsto 0$, and $ i\leftmapsto -i $
\end{proof}

Next we prove another combinatorial identity needed below.

\begin{lemma}\label{lemma:Combinatorics2}
    Let $n,q,z_1,z_2,t$ be natural numbers and let $j$ be an integer such that $0 \leq z_1,z_2\leq q \leq 2t$, $n\geq 2t$, and $z_1\geq z_2$. Define $\nb=n-2t+q$. Then,
    \begin{align} \label{eq: nbar}
        \frac{\binom{\nb}{j+z_1}\binom{\nb}{j+z_2}}{\binom{\nb+q}{j+q}} = \sum_{u=0}^{z_2}\sum_{v=0}^{q-z_1}\sum_{w=0}^{z_2-u}f_{z_1,z_2}(u,v,w)\binom{n-2t}{j-v-w+z_2},
    \end{align}
    where $f_{z_1,z_2}(u,v,w)$ does not depend on $j$.
\end{lemma}
\begin{proof}
Let us make a change of variables in \eqref{eq: nbar} by setting $j\leftmapsto j+z_2$ and $q\leftmapsto q-z_2$. Then we need to show that for $0\leq z_1-z_2\leq q\leq 2t-z_2$,
    \begin{align*}
        \frac{\binom{\nb}{j+(z_1-z_2)}\binom{\nb}{j}}{\binom{\nb+q+z_2}{j+q}} = \sum_{u=0}^{z_2}\sum_{v=0}^{q-(z_1-z_2)}\sum_{w=0}^{z_2-u}f_{z_1,z_2}(u,v,w)\binom{n-2t}{j-v-w}.
    \end{align*}
    Using Lemma \ref{Lemma:Combinatorics1}, we obtain
   \begin{align}\label{eq:CombLemma_1}
       \frac{\binom{\nb}{j+(z_1-z_2)}\binom{\nb}{j}}{\binom{\nb+q+z_2}{j+q}} = \frac{\binom{j+q}{j+(z_1-z_2)}\binom{\nb-j+z_2}{z_1}\binom{\nb}{j}}{\binom{q+z_2}{z_1}\binom{\nb+q+z_2}{q+z_2}}.
   \end{align}
Using the Vandermonde convolution, we obtain
   \begin{align}
       \binom{\nb-j+z_2}{z_1}=\binom{\nb-j+z_2}{\nb-j-(z_1-z_2)} &= \sum_{u=0}^{z_2}\binom{z_2}{u}\binom{\nb-j}{\nb-j-(z_1-z_2)-u}\nonumber \\ 
       &=\sum_{u=0}^{z_2}\binom{z_2}{u}\binom{\nb-j}{u+(z_1-z_2)}\label{eq:CombLemma_2}.
   \end{align}
Substituting \eqref{eq:CombLemma_2} into \eqref{eq:CombLemma_1}, we have
   \begin{align}
        \frac{\binom{\nb}{j+(z_1-z_2)}\binom{\nb}{j}}{\binom{\nb+q+z_2}{j+q}} &= \sum_{u=0}^{z_2}\frac{\binom{z_2}{u}}{\binom{q+z_2}{z_1}\binom{\nb+q+z_2}{q+z_2}}\binom{j+q}{j+(z_1-z_2)}\binom{\nb-j}{u+(z_1-z_2)}\binom{\nb}{j}\nonumber \\
        &=\sum_{u=0}^{z_2}\frac{\binom{z_2}{u}\binom{\nb}{u+(z_1-z_2)}}{\binom{q+z_2}{z_1}\binom{\nb+q+z_2}{q+z_2}}\binom{j+q}{j+(z_1-z_2)}\binom{\nb-(z_1-z_2)-u}{j}.\label{eq:CombLemma_3}
   \end{align}
   Setting $r=j, n=z_1-z_2, m=q-(z_1-z_2)$, and $l=\nb-q-u$ in Corollary \ref{Corr:Combinatorics}, we have the identity
   \begin{align}\label{eq:CombLemma_4}
       \binom{j+q}{j+(z_1-z_2)}\binom{\nb-(z_1-z_2)-u}{j}=\sum_{v=0}^{q-(z_1-z_2)}\binom{\nb-u+v}{v}\binom{q}{(z_1-z_2)+v}\binom{\nb-q-u}{j-v}. 
   \end{align}
Using this relation in \eqref{eq:CombLemma_3}, we obtain (remember the switch $q\leftmapsto q-z_2$)
   \begin{align}\label{eq:CombLemma_5}
       \frac{\binom{\nb}{j+(z_1-z_2)}\binom{\nb}{j}}{\binom{\nb+q+z_2}{j+q}} =\sum_{u=0}^{z_2}\sum_{v=0}^{q-(z_1-z_2)}\frac{\binom{z_2}{u}\binom{\nb}{u+(z_1-z_2)}\binom{\nb-u+v}{v}\binom{q}{(z_1-z_2)+v}}{\binom{q+z_2}{z_1}\binom{\nb+q+z_2}{q+z_2}}\binom{n-2t+(z_2-u)}{j-v}.
   \end{align}
   Consider the Vandermonde convolution
   \begin{align}\label{eq:CombLemma_6}
       \binom{n-2t+(z_2-u)}{j-v} = \sum_{w=0}^{z_2-u}\binom{z_2-u}{w}\binom{n-2t}{j-v-w}.
   \end{align}
   Combining \eqref{eq:CombLemma_5} with \eqref{eq:CombLemma_6}, we have
   \begin{align*}
        \frac{\binom{\nb}{j+(z_1-z_2)}\binom{\nb}{j}}{\binom{\nb+q+z_2}{j+q}} =\sum_{u=0}^{z_2}\sum_{v=0}^{q-(z_1-z_2)}\sum_{w=0}^{z_2-u}
        f_{z_1,z_2}(u,v,w)
        \binom{n-2t}{j-v-w},
   \end{align*}
   with
      $$
       f_{z_1,z_2}(u,v,w)=\frac{\binom{z_2}{u}\binom{\nb}{u+(z_1-z_2)}\binom{\nb-u+v}{v}\binom{q}{(z_1-z_2)+v}\binom{z_2-u}{w}}{\binom{q+z_2}{z_1}\binom{\nb+q+z_2}{q+z_2}},
       $$
   which completes the proof.
\end{proof}

To verify the KL conditions for the error set defined in \eqref{eq:DefErrorSet} and the codes given by Theorem~\ref{theorem:Main}, we will need to work with the explicit form of the Clebsch-Gordan coefficients
$C^{n/2-t+q,j-n/2+t}_{n/2,j+a-n/2;r,t-a}$. Below their upper and lower indices will always have this form,
where $n$ and $t$ are fixed and $r,a,q$ and $j$ may vary, so to shorten the notation, we will write 
   \begin{equation}\label{eq:Clebsch-short}
  C_{r,a}^q(j):=C^{n/2-t+q,j-n/2+t}_{n/2,j+a-n/2;r,t-a}.
   \end{equation}

In the following lemma we transform the classic expression for $C_{r,a}^q(j)$ to a form that enables the proof of the main theorem in the next section.
\begin{lemma}\label{lemma:Clebsch}
        Let $n,a,q,t,r$ be non-negative integers  such that $0\leq\ t-r\leq a,q\leq t+r\leq 2t$ and $n\geq 2t$, and let $j$ be an integer.  Define $\nb\coloneq n-2t+q$. The Clebsch-Gordan coefficient 
 can be written as
\remove{\begin{align*}
        C_{r,a}^q(j)=\begin{cases}
            \sum\limits_{k=t-r}^{q}\sum\limits_{k^\prime=0}^{t-r}h_{a,r}(k,k^\prime)\frac{\binom{\nb}{j+z}}{\sqrt{\binom{n}{j+a}\binom{\nb+q}{j+q}}}
             &\text{if }-\min(a,q)\leq j\leq n-\max(a,2t-q),\\
            0 &\text{otherwise},
        \end{cases}
\end{align*}
where $z=k-k^\prime$ and  $0\leq z\leq q$,  Here, the term 
\begin{align}\label{eq:LemmaClebsch_3}
    h_{a,r}(k,k^\prime)=
    \sqrt{\frac{\binom{n}{t+r-q}\binom{2r}{r+t-q}}{\binom{n+q+r-t+1}{r+t-q}\binom{2r}{a+r-t}}}(-1)^k\binom{q-(t-r)}{k-(t-r)}\binom{t+r-q}{a-k}\binom{\nb+(t-r)}{j+k}\binom{t-r}{k^\prime}
\end{align}
is a function of $k,k^\prime$, independent of $j$, and it is zero unless $z\leq a \leq 2t-q+z$.}

\begin{equation*}
        \begin{aligned}
            &C_{r,a}^q(j)=\sum\limits_{k=t-r}^{q}\sum\limits_{k^\prime=0}^{t-r}h_{a,r}(k,k^\prime)\frac{\binom{\nb}{j+z}}{\sqrt{\binom{n}{j+a}\binom{\nb+q}{j+q}}}
           &&\text{if }-\min(a,q)\leq j\leq n-\max(a,2t-q), \\
                          &C_{r,a}^q(j)= 0 &&\text{otherwise},
        \end{aligned}
        \end{equation*}
        where $z=k-k^\prime$ and  $0\leq z\leq q$,  Here, the term 
        \begin{align}\label{eq:LemmaClebsch_3}
        h_{a,r}(k,k^\prime)=
        \sqrt{\frac{\binom{n}{t+r-q}\binom{2r}{r+t-q}}{\binom{n+q+r-t+1}{r+t-q}\binom{2r}{a+r-t}}}(-1)^k\binom{q-(t-r)}{k-(t-r)}\binom{t+r-q}{a-k}\binom{\nb+(t-r)}{j+k}\binom{t-r}{k^\prime}
    \end{align}
    is a function of $k,k^\prime$, independent of $j$, and it is zero unless $z\leq a \leq 2t-q+z$.
\end{lemma}
\begin{proof} We start with an explicit expression of the Clebsch-Gordan coefficient $C^{J,m}_{j_1,m_1;j_2,m_2}$ as a binomial sum (see, e.g., \cite[Eq.27.9.1]{abramowitz1965handbook} or \cite[p.~240]{ClebcshBook}). For our values of the angular momentum, it is as follows:
    \begin{align}\label{eq:LemmaClebsch_1}
       C_{r,a}^q(j)=
        \sqrt{\frac{\binom{n}{t+r-q}\binom{2r}{r+t-q}}{\binom{n+q+r-t+1}{r+t-q}\binom{n}{j+a}\binom{2r}{a+r-t}\binom{\nb+q}{j+q}}}\sum_{k=t-r}^{q}(-1)^k\binom{q-(t-r)}{k-(t-r)}\binom{t+r-q}{a-k}\binom{\nb+(t-r)}{j+k}.
    \end{align}
By definition, $C_{r,a}^q(j)$ is zero outside of the region $-a\leq j\leq n-a$ and $-q\leq j\leq n-(2t-q)$  \cite[p.235]{ClebcshBook}. Hence, it gives a non-zero value only when $-\min(a,q)\leq j \leq n-\max(a,2t-q)$. 
Write $\binom{\nb+(t-r)}{j+k}$ as a Vandermonde convolution
    \begin{align}\label{eq:LemmaClebsch_2}
        \binom{\nb+(t-r)}{j+k} = \sum_{k^\prime=0}^{t-r}\binom{t-r}{k^\prime}\binom{\nb}{j+k-k^\prime}
    \end{align}
and substitute into \eqref{eq:LemmaClebsch_1} to obtain the expression in \eqref{eq:LemmaClebsch_3}.
    
Recalling that $t-r\leq k \leq q$ and $-(t-r)\leq -k^\prime \leq 0$, we have $0\leq z=k-k^\prime \leq q$. Furthermore, the term $\binom{t+r-q}{a-k}$ on the RHS of \eqref{eq:LemmaClebsch_3} equals zero unless $k\leq a \leq t+r-q+k$. Since $k=k^\prime+z$ and $0\leq k^\prime\leq t-r$, we have $z\leq z+k^\prime\leq a \leq t+r-q+z+k^\prime\leq 2t-q+z$.
\end{proof}

\section{Proof of Theorem \ref{theorem:Main}}\label{sec:main thm proof}
Conditions (C1) and (C2) are required to ensure that codewords $\ket{\bfc_0}$ and $\ket{\bfc_1}$ form orthonormal states. For the conditions (C3) and (C4), we will analyze the KL conditions for the error set defined in \eqref{eq:DefErrorSet}. In other words, our aim is to show that relations
    \begin{align*}
        &\bra{\bfc_1}(\hat{E}^{r_2,\delta J_2}_{\delta m_2})^\dagger\hat{E}^{r_1,\delta J_1}_{\delta m_1}\ket{\bfc_0}=0,\\
        &\bra{\bfc_0}(\hat{E}^{r_2,\delta J_2}_{\delta m_2})^\dagger\hat{E}^{r_1,\delta J_1}_{\delta m_1}\ket{\bfc_0}-\bra{\bfc_1}(\hat{E}^{r_2,\delta J_2}_{\delta m_2})^\dagger\hat{E}^{r_1,\delta J_1}_{\delta m_1}\ket{\bfc_1}=0
    \end{align*}
    hold for all $\hat{E}^{r_2,\delta J_2}_{\delta m_2}, \hat{E}^{r_1,\delta J_1}_{\delta m_1}\in \cE_t$, $-r_1\leq \delta m_1,\delta J_1\leq r_1\leq t$ and $-r_2\leq \delta m_2,\delta J_2\leq r_2 \leq t$ whenever conditions (C3) and (C4) are satisfied. 
    
    Let us define parameters
    \begin{gather*}
  a^\prime= t-\delta m_1, \; q_1= t+\delta J_1,\;
   b^\prime= t-\delta m_2,\;  q_2= t+\delta J_2,
   \end{gather*}
and make a change of variables $j\leftmapsto j-a, j'\leftmapsto j'-b.$

Note that acting by $E_{\delta m}^{r,\delta J}$ on $\ket {J,m}$ results in
   $$
E_{\delta m}^{r,\delta J}\ket {J,m}\propto C_{J,m;r_1,\delta m}^{J+\delta J,m+\delta m}
\ket{J+\delta J, m+\delta m}.
  $$
Therefore, with the notation introduced above, the short form of the Clebsch-Gordan coefficients \eqref{eq:Clebsch-short}, and the definition of $\ket{\bfc_i}, i=0,1$ in  \eqref{eq:DefCodeSpace},
we obtain
    \begin{align*}
        &\hat{E}^{r_1,\delta J_1}_{\delta m_1}\ket{\bfc_0}\propto \sum_j 
        C_{r_1,a'}^{q_1}(j)
        \alpha_{j+a^\prime}\;\ket{n/2-t+q_1,j-n/2+t},\\
        &\hat{E}^{r_2,\delta J_2}_{\delta m_2}\ket{\bfc_1}\propto \sum_{j^\prime} 
        C_{r_2,b'}^{q_2}(j')
        \beta_{j^\prime+b^\prime}\;\ket{n/2-t+q_2,j^\prime-n/2+t}.
    \end{align*}
    Hence, the inner product
    \begin{align*}
        \bra{\bfc_1}(\hat{E}^{r_2,\delta J_2}_{\delta m_2})^\dagger\hat{E}^{r_1,\delta J_1}_{\delta m_1}\ket{\bfc_0}=\sum_j\sum_{j^\prime}
        C_{r_1,a'}^{q_1}(j)C_{r_2,b'}^{q_2}(j')
        \alpha_{j+a^\prime}\beta_{j^\prime+b^\prime}\delta_{q_1,q_2}\delta_{j,j^\prime}.
    \end{align*}
    Setting $q=q_1=q_2$, $\nb=n-2t+q$, using the Lemma \ref{lemma:Clebsch}, and changing the summation index we obtain
    \begin{multline}\label{eq:Proof_1}
        \bra{\bfc_1}(\hat{E}^{r_2,\delta J_2}_{\delta m_2})^\dagger\hat{E}^{r_1,\delta J_1}_{\delta m_1}\ket{\bfc_0}=
        \sum_{k_1=t-r_1}^{q}\sum_{k_1^\prime=0}^{t-r_1}\sum_{k_2=t-r_2}^{q}\sum_{k_2^\prime=0}^{t-r_2}h_{a^\prime,r_1}(k_1,k_1^\prime)h_{b^\prime,r_2}(k_2,k_2^\prime)\\
        \times\left(\sum_{j=-\min(a^\prime,b^\prime,q)}^{n-\max(a^\prime,b^\prime,2t-q)}\frac{1}{\sqrt{\binom{n}{j+a^\prime}\binom{n}{j+b^\prime}}}\frac{\binom{\nb}{j+z_1}\binom{\nb}{j+z_2}}{\binom{\nb+q}{j+q}}\alpha_{j+a^\prime}\beta_{j+b^\prime}\right).
    \end{multline}
    Hence it is sufficient to show that for all $z_1,z_2\leq q \leq 2t$, $z_1\leq a^\prime \leq 2t-q+z_1$, and $z_2\leq b^\prime \leq 2t-q+z_2$, the sum inside the parentheses is zero whenever condition (C3) is met. W.L.O.G we can assume that $z_1\geq z_2$. Using Lemma \ref{lemma:Combinatorics2} and changing the index of the summation, we obtain
    \begin{multline}\label{eq:Proof_2}
        \sum_{j=-\min(a^\prime,b^\prime,q)}^{n-\max(a^\prime,b^\prime,2t-q)}\frac{1}{\sqrt{\binom{n}{j+a^\prime}\binom{n}{j+b^\prime}}}\frac{\binom{\nb}{j+z_1}\binom{\nb}{j+z_2}}{\binom{\nb+q}{j+q}}\alpha_{j+a^\prime}\beta_{j+b^\prime}=\\
       \sum_{u=0}^{z_2}\sum_{v=0}^{q-z_1}\sum_{w=0}^{z_2-u}f_{z_1,z_2}(u,v,w) \left(\sum_{j=-\min(a^\prime,b^\prime,q)}^{n-\max(a^\prime,b^\prime,2t-q)}\frac{\binom{n-2t}{j-v-w+z_2}}{\sqrt{\binom{n}{j+a^\prime}\binom{n}{j+b^\prime}}}\alpha_{j+a^\prime}\beta_{j+b^\prime}\right).
    \end{multline}
Make a variable change $j\leftmapsto j-v-w+z_2$. Then the sum inside the parentheses in \eqref{eq:Proof_2} becomes
    \begin{align}\label{eq:Proof_3}
        \sum_{j_{\min}\le j\le j_{\max}}
        \frac{\binom{n-2t}{j}}{\sqrt{\binom{n}{j+a^\prime+v+w-z_2}\binom{n}{j+b^\prime+v+w-z_2}}}\alpha_{j+a^\prime+v+w-z_2}\beta_{j+b^\prime+v+w-z_2},
    \end{align}
where 
   $$j_{\min}=-\min(a^\prime,b^\prime,q)-v-w+z_2\text{ and }j_{\max}=n-\max(a^\prime,b^\prime,2t-q)-v-w+z_2.
   $$
We claim that it is possible to set $j_{\min}$ to 0 and $j_{\max}$ to $n$ without changing the value of the sum.
Let us begin with $j_{\min}$. Since $a^\prime\geq z_1$, $b^\prime \geq z_2$, and $v,w\ge0$, 
we conclude that $j_{\min}\leq \max(0,z_2-z_1,z_2-q)$. Since $z_1\geq z_2$ and $z_2\leq q$, this implies that $j_{\min}\leq 0$. 
Finally, since the negative values of $j$ do not contribute to the sum, we can replace $j_{\min}$ with 0. Turning to
$j_{\max}$, recall that $v\leq q-z_1$, $w\leq z_2-u\leq z_2$, $a^\prime \leq 2t-q+z_1$, and $b^\prime \leq z_2+2t-q$.
Therefore, $j_{\max}\geq n-2t+\min(0,z_2,z_1-z_2)\geq n-2t$. Since the values of $j$ in excess of $n-2t$ do not
contribute to the sum, we can set $j_{\max}=n$. 

Next let $a=a^\prime + v + w -z_2$, then 
    $$
    0\leq z_1-z_2\leq a \leq 2t-q+z_1+q-z_1+z_2-u-z_2=2t-u\leq 2t.
    $$
Similarly, let $b=b^\prime+v+w-z_2$, then 
   $$
   0\leq b \leq 2t-q+z_2+q-z_1+z_2-u-z_2\leq 2t-(z_1-z_2)-u\leq 2t.
   $$
Collecting this information, we have for the inner sum in \eqref{eq:Proof_2}
\begin{align}\label{eq:Proof_4}
    \sum_{j=-\min(a^\prime,b^\prime,q)}^{n-\max(a^\prime,b^\prime,2t-q)}\frac{\binom{n-2t}{j-v-w+z_2}}{\sqrt{\binom{n}{j+a^\prime}\binom{n}{j+b^\prime}}}\alpha_{j+a^\prime}\beta_{j+b^\prime}=    \sum_{j=0}^{n}\frac{\binom{n-2t}{j}}{\sqrt{\binom{n}{j+a}\binom{n}{j+b}}}\alpha_{j+a}\beta_{j+b},
\end{align}
where $a$ and $b$ are some integers such that $0\leq a,b \leq 2t$. Combining \eqref{eq:Proof_1}, \eqref{eq:Proof_2}, and \eqref{eq:Proof_4}, we can write \eqref{eq:Proof_1} as
\begin{multline}\label{eq:inner}
     \bra{\bfc_1}(\hat{E}^{r_2,\delta J_2}_{\delta m_2})^\dagger\hat{E}^{r_1,\delta J_1}_{\delta m_1}\ket{\bfc_0}=\\
        \sum_{k_1=t-r_1}^{q}\sum_{k_1^\prime=0}^{t-r_1}\sum_{k_2=t-r_2}^{q}\sum_{k_2^\prime=0}^{t-r_2}\sum_{u=0}^{z_2}\sum_{v=0}^{q-z_1}\sum_{w=0}^{z_2-u}h_{a^\prime,r_1}(k_1,k_1^\prime)h_{b^\prime,r_2}(k_2,k_2^\prime)
        f_{z_1,z_2}(u,v,w)\\
        \times
        \left( \sum_{j=0}^{n}\frac{\binom{n-2t}{j}}{\sqrt{\binom{n}{j+a}\binom{n}{j+b}}}\alpha_{j+a}\beta_{j+b}\right)
\end{multline}
for some values $0\leq a,b \leq 2t$. 
Engaging condition (C3) of the theorem, we now claim that 
   $$
   \bra{\bfc_1}(\hat{E}^{r_2,\delta J_2}_{\delta m_2})^\dagger\hat{E}^{r_1,\delta J_1}_{\delta m_1}\ket{\bfc_0}=0
   $$
for all $\hat{E}^{r_2,\delta J_2}_{\delta m_2},\hat{E}^{r_1,\delta J_1}_{\delta m_1}\in \cE_t$.

Following the same sequence of steps, one can show that, assuming condition (C4) of the theorem, 
   \begin{equation}\label{eq:C4}
\bra{\bfc_0}(\hat{E}^{r_2,\delta J_2}_{\delta m_2})^\dagger\hat{E}^{r_1,\delta J_1}_{\delta m_1}\ket{\bfc_0}=\bra{\bfc_1}(\hat{E}^{r_2,\delta J_2}_{\delta m_2})^\dagger\hat{E}^{r_1,\delta J_1}_{\delta m_1}\ket{\bfc_1}
  \end{equation}
for all $\hat{E}^{r_2,\delta J_2}_{\delta m_2},\hat{E}^{r_1,\delta J_1}_{\delta m_1}\in \cE_t$. Indeed,
it is by now clear that the expression for $\bra{\bfc_i}(\hat{E}^{r_2,\delta J_2}_{\delta m_2})^\dagger\hat{E}^{r_1,\delta J_1}_{\delta m_1}\ket{\bfc_i},$ $ i=0,1$ takes the form similar to the one derived above in \eqref{eq:inner}, except for the fact that in the sum
in the parentheses we will obtain $\alpha_{j+a}\alpha_{j+b}$ for $i=0$ and $\beta_{j+a}\beta_{j+b}$ for $i=1$. These
expressions will coincide, and thus \eqref{eq:C4} will be fulfilled, if condition (C4) holds true.

Finally, to prove the error detection conditions, we need to show that
\begin{align*}
        0&=\bra{\bfc_1}\hat{E}^{r,\delta J}_{\delta m}\ket{\bfc_0},\\
        0&=\bra{\bfc_0}\hat{E}^{r,\delta J}_{\delta m}\ket{\bfc_0}-\bra{\bfc_1}\hat{E}^{r,\delta J_1}_{\delta m}\ket{\bfc_1}
    \end{align*}
for all $r\leq t$. This follows upon making a change of variables $2t\rightarrow t$. 

\section{Offset supports and the construction of \cite{AEcodes}}\label{sec: example}
The proposal in \cite{AEcodes} is designed around eliminating the off-diagonal terms in the inner products appearing in the
KL condition. Starting with the code definition \eqref{eq:DefCodeSpace}, 
define the supports of the basis vectors as $\supp(\ket{\bfc_0})=\{j : \alpha_j\neq 0\}$ and $\supp(\ket{\bfc_1})=\{j : \beta_j\neq 0\}$. Assume that they are staggered in the sense that
     $$
     |j_k-j_l|\geq 2t+1 \quad \text{ for all }j_k,j_l\in \supp(\ket{\bfc_0})\cup \supp(\ket{\bfc_1}).
     $$
We quickly observe that now condition (C3) is satisfied trivially since $\alpha_{j+a}\beta_{j+b}=0$ for all $j=1,2,\ldots,n$ and $0\leq a,b \leq 2t$. 
To verify condition (C4), it is sufficient to examine the diagonal terms $a=b$. We obtain that it holds if
    \begin{align*}
        \sum_{j=0}^n\frac{\binom{n-2t}{j}}{\binom{n}{j+a}}\left(\alpha_{j+a}^2-\beta_{j+a}^2\right)=0
    \end{align*}
for all $0\leq a \leq 2t$. Rewriting the left-hand side based on Lemma \ref{Lemma:Combinatorics1}, we further obtain
    \begin{align*}
        \frac{1}{\binom{2t}{a}\binom{n}{2t}}\sum_{j=0}^n\binom{j+a}{a}\binom{n-j-a}{2t-a}\left(\alpha_{j+a}^2-\beta_{j+a}^2\right) = 0.
    \end{align*}
 Upon the variable change $j\leftmapsto j+a$, both conditions (C3) and (C4) are satisfied if and only if
    \begin{align*}
        \sum_{j=0}^n\binom{j}{a}\binom{n-j}{2t-a}\left(\alpha_{j}^2-\beta_{j}^2\right) = 0 \quad\quad \text{for all $0\leq a\leq 2t$}.
    \end{align*}
The term term $\binom{j}{a}\binom{n-j}{2t-a}$ is a polynomial in $j$ of degree $2t$. Hence, if
    \begin{align}\label{eq:diagonal}
        \sum_{j=0}^n \left( \alpha_j^2-\beta_j^2 \right)j^i=0 \quad\quad \text{for all $0\leq i\leq 2t$}
    \end{align}
conditions (C3) and (C4) hold true. Thus, if we manage to solve this system of equations for 
$\alpha_j,\beta_j$, this will yield an AE code with the basis given in \eqref{eq:DefCodeSpace}. This recovers 
the construction strategy of AE codes in \cite{AEcodes}, see esp. their Eq.~(C25). 

We note that, while finding solutions to system \eqref{eq:diagonal} is not too hard, it becomes a difficult task
once we attempt to construct more efficient codes. Eliminating the off-diagonal conditions gives a way of approaching this, which here is achieved by inserting gaps between the excitation levels. See also \cite{aydin2023family} for other ways of addressing this issue.

\section{Proof of Proposition \ref{prop:PICodes}}
Let $\cC$ be a $k$-dimensional permutation-invariant code with basis $\{\ket{\bfc_0},\ket{\bfc_1},\ldots,\ket{\bfc_k}\}$, and $\cQ$ be the corresponding AE code defined by the basis $\{e\ket{\bfc_i}):i=1,2,\ldots,k\}$. Let us choose two arbitrary basis vectors
    \begin{align*}
        \ket{\bfc_l} = \sum_{j=0}^n\alpha_j\ket{D^n_j}\quad\text{and}\quad \ket{\bfc_m}=\sum_{j=0}^n\beta_j\ket{D^n_j},
    \end{align*}
    where $\alpha_j,\beta_j$ are reals for all $j=1,2,\ldots,n$. Then the corresponding codewords of $\cQ$ has the form,
    \begin{align*}
        \ket{\tilde{\bfc_l}}&=e\left(\ket{\bfc_l}\right)=\sum_{j=0}^n\alpha_j\ket{n/2,j-n/2},\\
        \ket{\tilde{\bfc_m}}&=e\left(\ket{\bfc_m}\right)=\sum_{j=0}^n\beta_j\ket{n/2,j-n/2}.
    \end{align*}
Suppose $\cC$ corrects $t$ errors, so the coefficient vectors of any pair of basis codewords must satisfy the conditions of
Theorem \ref{thm:PIConditions}. This implies that the coefficient vectors $\boldsymbol{\alpha}= (\alpha_0,\alpha_1,\ldots,\alpha_n) $ and $ \boldsymbol{\beta} = (\beta_0,\beta_1,\ldots,\beta_n) $ also satisfy them in the sense that these conditions hold for any
choice of pairs of $\alpha$'s and/or $\beta$'s. Now recall that the conditions on the coefficients in Theorem \ref{thm:PIConditions} 
are the same as in Theorem \ref{theorem:Main}. Referring to Theorem \ref{theorem:ErrorCorrection}, we conclude that the code $\cQ$ corrects up to order-$t$ transitions and detects up to order-$2t$ transitions. 

\twocolumngrid

\bibliography{PI}

\end{document}